\documentclass[ssy,preprint]{imsart_axv}

\usepackage{booktabs} 
\usepackage[active]{srcltx}
\usepackage{color}
\usepackage{verbatim}
\usepackage{amssymb}
\usepackage{graphicx}
\usepackage{breakurl}

\RequirePackage[OT1]{fontenc}
\RequirePackage{amsthm,amsmath}
\RequirePackage[numbers]{natbib}
\RequirePackage[colorlinks,citecolor=blue,urlcolor=blue]{hyperref}

\startlocaldefs
\numberwithin{equation}{section}
\theoremstyle{plain}
\newtheorem{thm}{\protect\theoremname}[section]
\theoremstyle{definition}
\newtheorem{defn}[thm]{\protect\definitionname}
\theoremstyle{remark}
\newtheorem{rem}[thm]{\protect\remarkname}
\theoremstyle{plain}
\newtheorem{prop}[thm]{\protect\propositionname}
\theoremstyle{plain}
\newtheorem{lem}[thm]{\protect\lemmaname}

\makeatother

\providecommand{\definitionname}{Definition}
\providecommand{\lemmaname}{Lemma}
\providecommand{\propositionname}{Proposition}
\providecommand{\remarkname}{Remark}
\providecommand{\theoremname}{Theorem}

\global\long\def\Real{\mathbb{R}}
\global\long\def\E{\mathbb{E}}
\global\long\def\Int{\mathbb{Z}}
\global\long\def\Indicator{\mathbb{I}}
\global\long\def\n{\boldsymbol{n}}

\global\long\def\N{\boldsymbol{N}}
\global\long\def\Y{\boldsymbol{Y}}

\global\long\def\m{\boldsymbol{m}}

\global\long\def\blambda{\boldsymbol{\lambda}}
\global\long\def\e{\boldsymbol{e}}
\global\long\def\bphi{\boldsymbol{\phi}}
\global\long\def\bpi{\boldsymbol{\pi}}

\global\long\def\balpha{\boldsymbol{\alpha}}
\global\long\def\Q{\boldsymbol{Q}}
\global\long\def\R{\boldsymbol{R}}
\global\long\def\Z{\boldsymbol{Z}}
\global\long\def\Nc{\mathcal{N}_{C}^{(\epsilon)}}
\global\long\def\Nd{\mathcal{N}_{D}^{(\epsilon)}}
\global\long\def\xe{x^{(\epsilon)}}
\newcommand{\dominate}{\preccurlyeq}
\endlocaldefs

\begin{document}
	
\begin{frontmatter}	
\title{Centralized Congestion Control and Scheduling in a Datacenter}
\runtitle{Centralized Congestion Control and Scheduling}

\begin{aug}
	\author{\fnms{Devavrat} \snm{Shah}\thanksref{t1}\ead[label=e1]{devavrat@mit.edu }}
	\address{\printead{e1}}
	\and
	\author{\fnms{Qiaomin} \snm{Xie}\thanksref{t1}\ead[label=e2]{qxie@mit.edu}}
	\address{\printead{e2}}
	\runauthor{D. Shah and Q. Xie}
	\affiliation{Massachusetts Institute of Technology}
\end{aug}



\begin{abstract}

We consider the problem of designing a packet-level congestion control
and scheduling policy for datacenter networks. Current datacenter
networks primarily inherit the principles that went into the design of Internet,
where congestion control and scheduling are distributed. While distributed
architecture provides robustness, it suffers in terms of performance. 
Unlike Internet, data center is fundamentally a ``controlled'' 
environment. This raises the possibility of designing a centralized architecture
to achieve better performance. Recent solutions such as Fastpass \cite{perry2014fastpass}
and Flowtune \cite{perry17flowtune} have provided the proof of this concept. This
raises the question: what is theoretically optimal performance achievable 
in a data center?

We propose a centralized policy that guarantees a per-flow end-to-end
flow delay bound of $O$(\#hops $\times$ flow-size $/$ gap-to-capacity). 
Effectively such an end-to-end delay will be experienced by flows even if we
removed congestion control and scheduling constraints as the resulting 
queueing networks can be viewed as the classical {\em reversible} multi-class 
queuing network, which has a product-form stationary distribution. In the language
of \cite{harrison2014bandwidth}, we establish that {\em baseline} performance for
this model class is achievable.

Indeed, as the key contribution of this work, we propose a method to
{\em emulate} such a reversible queuing network while satisfying congestion
control and scheduling constraints. Precisely, our policy is an emulation
of Store-and-Forward (SFA) congestion control in conjunction with  
Last-Come-First-Serve Preemptive-Resume (LCFS-PR) scheduling policy. 

\end{abstract}




\end{frontmatter}

\section{Introduction}

With an increasing variety of applications and workloads being hosted
in datacenters, it is highly desirable to design datacenters that
provide high throughput and low latency. Current datacenter networks
primarily employ the design principle of Internet, where congestion
control and packet scheduling decisions are distributed among endpoints
and routers. While distributed architecture provides scalability and
fault-tolerance, it is known to suffer from throughput loss and high
latency, as each node lacks complete knowledge of entire network conditions
and thus fails to take a globally optimal decision. 

The datacenter network is fundamentally different from the wide-area
Internet in that it is under a single administrative control. Such
a single-operator environment makes a centralized architecture a feasible option.
Indeed, there have been recent proposals for centralized control
design for data center networks \cite{perry17flowtune,perry2014fastpass}.
In particular, Fastpass \cite{perry2014fastpass} uses a centralized
arbiter to determine the path as well as the time slot of transmission
for each packet, so as to achieve zero queueing at switches. Flowtune
\cite{perry17flowtune} also uses a centralized controller, but congestion
control decisions are made at the granularity of a flowlet, with the
goal of achieving rapid convergence to a desired rate allocation.
Preliminary evaluation of these approaches demonstrates promising
empirical performance, suggesting the feasibility of a centralized design
for practical datacenter networks.

Motivated by the empirical success of the above work, we are interested
in investigating the theoretically optimal performance achievable by a centralized
scheme for datacenters. Precisely, we consider a centralized architecture,
where the congestion control and packet transmission are delegated
to a centralized controller. The controller collects all dynamic endpoints
and switches state information, and redistributes the congestion control
and scheduling decisions to all switches/endpoints. We propose a packet-level
policy that guarantees a per-flow end-to-end flow delay bound of $\ensuremath{O}(\#\text{hops}\ensuremath{\times}\text{flow-size}\ensuremath{/}\text{gap-to-capacity})$.
To the best of our knowledge, our result is the first one to show
that it is possible to achieve such a delay bound in a network with
congestion control and scheduling constraints. Before describing the
details of our approach, we first discuss related work addressing
various aspects of the network resource allocation problem. 

\subsection{Related Work}

There is a very rich literature on congestion control and scheduling. The literature
on congestion control has been primarily driven by bandwidth allocation in the
context of Internet. The literature on packet scheduling has been historically 
driven by managing supply-chain (multi-class queueing networks), telephone
networks (loss-networks), switch / wireless networks (packet
switched networks) and now data center networks. In what follows, we 
provide brief overview of representative results from theoretical and 
systems literature.

\medskip
\noindent\textbf{Job or Packet Scheduling:} A scheduling policy in the context 
of classical multi-class queueing networks essentially specifies the service 
discipline at each queue, i.e., the order in which waiting jobs are served. Certain
service disciplines, including the last-come-first-serve preemptive-resume
(LCFS-PR) policy and the processor sharing discipline, are known to
result in quasi-reversible multi-class queueing networks, which have
a product form equilibrium distribution \cite{kelly1979reversibility}. The crisp 
description of the equilibrium distribution makes these disciplines remarkably 
tractable analytically. 

More recently, the scheduling problems for switched networks, which
are special cases of stochastic processing networks as introduced
by Harrison \cite{harrison2000brownian}, have attracted a lot of
attention starting \cite{tassiulas1992maxweight} including some recent
examples \cite{walton2014concave,shah2014SFA,maguluri2015heavy}. Switched networks are 
queueing networks where there are constraints on which queues can be served 
simultaneously. They effectively model a variety of interesting applications, 
exemplified by wireless communication networks, and input-queued switches for 
Internet routers. The MaxWeight/BackPressure policy, introduced by 
Tassiulas and Ephremides for wireless communication \cite{tassiulas1992maxweight, mckeown1996achieving}, 
have been shown to achieve a maximum throughput stability for switched networks. 
However, the provable delay bounds of this scheme scale with the number of queues in the
network. As the scheme requires maintaining one queue per route-destination at each one,
the scaling can be potentially very bad. For instance, recently Gupta
and Javidi \cite{gupta2007routing} showed that such an algorithm
can result in very poor delay performance via a specific example.
Walton \cite{walton2014concave} proposed a proportional scheduler
which achieves throughput optimality as the BackPressure policy,
while using a much simpler queueing structure with one queue per link.
However, the delay performance of this approach is unknown. Recently
Shah, Walton and Zhong \cite{shah2014SFA} proposed a policy where
the scheduling decisions are made to approximate a queueing network
with a product-form steady state distribution. The policy achieves optimal
queue-size scaling for a class of switched networks.In a recent work, Theja and Srikant \cite{maguluri2015heavy} established heavy-traffic optimality of
MaxWeight policy for input-queued switches. 

\medskip
\noindent\textbf{Congestion control:} A long line of literature on congestion
control began with the work of Kelly, Maulloo and Tan \cite{kelly1998rate}, where they
introduced an optimization framework for flow-level resource allocation in
the Internet. In particular, the rate control algorithms are developed
as decentralized solutions to the utility maximization problems. The
utility function can be chosen by the network operator to achieve
different bandwidth and fairness objectives. Subsequently, this optimization
framework has been applied to analyze existing congestion control
protocols (such as TCP) \cite{low2002vegas,low2002internet,mo2000fair}; a
comprehensive overview can be found in  \cite{srikant_book}.
Roberts and Massouli\'{e} \cite{massoulie2000bandwidth} applied this
paradigm to settings where flows stochastically depart and arrive,
known as bandwidth sharing networks. The resulting proportional fairness
policies have been shown to be maximum stable \cite{bonald2001fairness,massoulie2007fairness}.
The heavy traffic behavior of proportional fairness has been subsequently
studied \cite{shah2014qualitative,kang2009diffusion}. Another bandwidth
allocation of interest is the store-and-forward allocation (SFA) policy,
which was first introduced by Massouli\'{e} (see Section 3.4.1 in~\cite{proutiere_thesis}) and
later analyzed in the thesis of Prouti\`{e}re \cite{proutiere_thesis}. The SFA policy has the remarkable
property of insensitivity with respect to service distributions,
as shown by Bonald and Prouti\`{e}re \cite{bonald2003insensitive}, and
Zachary\cite{zachary2007insensitivity}. Additionally, this policy
induces a product-form stationary distribution \cite{bonald2003insensitive}.
The relationship between SFA and proportional fairness has been explored
\cite{massoulie2007fairness}, where SFA was shown to converge to
proportional fairness with respect to a specific asymptote. 

\medskip
\noindent\textbf{Joint congestion control and scheduling:} More recently, the
problem of designing joint congestion-control and scheduling mechanisms
has been investigated \cite{eryilmaz2006joint,lin2004joint,stolyar2005maximizing}.
The main idea of these approaches is to combine a queue-length-based
scheduler and a distributed congestion controller developed for wireline
networks, to achieve stability and fair rate allocation. For instance,
the joint scheme proposed by Eryilmaz and Srikant combines the BackPressure
scheduler and a primal-dual congestion controller for wireless networks.
This line of work focuses on addressing the question of the stability. The
Lyapunov function based delay (or queue-size) bound for such algorithm are
relatively very poor. It is highly desirable to design a joint mechanism that is provably
throughput optimal and has low delay bound. Indeed, the work of Moallemi and Shah \cite{moallemi2010flow} was an attempt in this direction, where they
developed a stochastic model that jointly captures the packet- and 
flow-level dynamics of a network, and proposed a joint policy based on 
$\alpha$-weighted policies. They argued that in a certain asymptotic regime 
(critically loaded fluid model) the resulting algorithm induces queue-sizes that are
within constant factor of optimal quantities. However, this work stops short of 
providing non-asymptotic delay guarantees.

\medskip
\noindent\textbf{Emulation:} In our approach we utilize the concept of emulation,
which was introduced by Prabhakar and Mckeown \cite{prabhakar1999speedup}
and used in the context of bipartite matching. Informally, a network
is said to emulate another network, if the departure processes from
the two networks are identical under identical arrival processes.
This powerful technique has been subsequently used in a variety of
applications \cite{jagabathula2008delay_scheduling,shah2014SFA,gamal2006throughput_delay,chuang1999matching}.
For instance, Jagabathula and Shah designed a delay optimal scheduling
policy for a discrete-time network with arbitrary constraints, by
emulating a quasi-reversible continuous time network \cite{jagabathula2008delay_scheduling};
The scheduling algorithm proposed by Shah, Walton and Zhong \cite{shah2014SFA}
for a single-hop switched network, effectively emulates the bandwidth sharing
network operating under the SFA policy. However, it is unknown how
to apply the emulation approach to design a joint congestion control
and scheduling scheme.

\medskip
\noindent\textbf{Datacenter Transport:}
Here we restrict to system literature in the context of datacenters.
Since traditional TCP developed for wide-area Internet does not meet
the strict low latency and high throughput requirements in datacenters,
new resource allocation schemes have been proposed and deployed \cite{alizadeh2010DCTCP,alizadeh2013pfabric,nagaraj2016numfabric,hong2012pdq,perry17flowtune,perry2014fastpass}.
Most of these systems adopt distributed congestion control schemes,
with the exception of Fasspass \cite{perry2014fastpass} and Flowtune
\cite{perry17flowtune}.

DCTCP \cite{alizadeh2010DCTCP} is a delay-based (queueing) congestion
control algorithm with a similar control protocol as that in TCP. It aims to keep the switch queues small, by leveraging Explicit Congestion Notification (ECN) to provide multi-bit feedback to the end points. Both pFabric \cite{alizadeh2013pfabric} and
PDQ \cite{hong2012pdq} aim to reduce flow completion time, by utilizing
a distributed approximation of the shortest remaining flow first policy.
In particular, pFabric uses in-network packet scheduling to decouple
the network's scheduling policy from rate control. NUMFabric \cite{nagaraj2016numfabric}
is also based on the insight that utilization control and network
scheduling should be decoupled. In particular, it combines a packet
scheduling mechanism based on weighted fair queueing (WFQ) at the
switches, and a rate control scheme at the hosts that is based on
the network utility maximization framework.

Our work is motivated by the recent successful stories that demonstrate
the viability of centralized control in the context of datacenter
networks \cite{perry2014fastpass,perry17flowtune}. Fastpass \cite{perry2014fastpass}
uses a centralized arbiter to determine the path as well as the time
slot of transmission for each packet. To determine the set of sender-receiver
endpoints  that can communicate in a timeslot, the arbiter views
the entire network as a single input-queued switch, and uses a heuristic
to find a matching of endpoints in each timeslot. The arbiter then
chooses a path through the network for each packet that has been allocated
timeslots. To achieve zero-queue at switches, the arbiter assigns
packets to paths such that no link is assigned multiple packets in
a single timeslot. That is, each packet is arranged to arrive at a
switch on the path just as the next link to the destination becomes
available. 

Flowtune \cite{perry17flowtune} also uses a centralized controller,
but congestion control decisions are made at the granularity of a
flowlet, which refers to a batch of packets backlogged at a sender.
It aims to achieve fast convergence to optimal rates by avoiding packet-level
rate fluctuations. To be precise, a centralized allocator computes
the optimal rates for a set of active flowlets, and those rates are
updated dynamically when flowlets enter or leave the network. In particular,
the allocated rates maximize the specified network utility, such
as proportional fairness.

\medskip
\noindent\textbf{Baseline performance:} In a recent work Harrison et al.~\cite{harrison2014bandwidth} studied the \emph{baseline performance} for congestion control, that is, an achievable benchmark for the delay performance in flow-level models. Such a benchmark provides an upper bound on the optimal achievable performance. In particular, baseline performance in flow-level models is exactly achievable by the store-and-forward allocation (SFA) mentioned earlier. On the other hand, the work by Shah et al.~\cite{shah2014SFA} established baseline performance for scheduling in packet-level networks. They proposed a scheduling policy that effectively emulates the bandwidth sharing network under the SFA policy. The results for both flow- and packet-level models boil down to a product-form stationary distribution, where each component of the product-form behaves like an $M/M/1$ queue. However, no baseline performance has been established for a hybrid model with flow-level congestion control and packet scheduling. 

\smallskip

This is precisely the problem we seek to address in this paper.
The goal of this paper is to understand what is the best performance achievable by centralized designs in datacenter networks. In particular, we aim to establish baseline performance for datacenter networks with congestion control and scheduling constraints. 
To investigate this problem, we consider a datacenter network with a tree topology, and focus on a hybrid model with simultaneous dynamics of flows and packets. Flows arrive at each endpoint according to an exogenous process and wish to transmit some amount of data through the network. As in standard congestion control algorithms, the flows generate packets at their ingress to the network. The packets travel to their respective destinations along links in the network. We defer the model details to Section \ref{sec:Model-and-Notation}.

\subsection{Our approach}

The control of a data network comprises of two sub-problems: congestion control and scheduling. 
On the one hand, congestion control aims to ensure fair sharing of network resources among endpoints and to minimize congestion inside the network. The congestion control policy determines the rates at which each endpoint injects data into the internal network for transmission. On the other hand, the internal network maintains buffers for packets that are in transit across the network, where the queues at each buffer are managed according to some packet scheduling policy.

Our approach addresses these two sub-problems simultaneously, with the overall architecture shown in Figure \ref{fig:Overview}. The system decouples congestion control and in-network packet scheduling by maintaining two types of buffers: \emph{external} buffers which store arriving flows of different types, and \emph{internal} buffers for packets in transit across the network. In particular, there is a separate external buffer for each type of arriving flows. Internally, at each directed link~$l$ between two nodes $(u,v)$ of the network, there is an internal buffer for storing packets waiting at node~$u$ to pass through link~$l$. Conceptually, the internal queueing structure 
corresponds to the output-queued switch fabric of ToR and core switches in a datacenter, so each directed link is abstracted as a queueing server for packet transmission.

Our approach employs independent mechanisms for the two sub-problems. The congestion control policy uses only the state of external buffers for rate allocation, and is hence decoupled from packet scheduling. The rates allocated for a set of flows that share the network will change only when new flows arrive or when flows are admitted into the internal network for transmission. For the internal network, we adopt a packet scheduling mechanism based on the dynamics of internal buffers.

\begin{figure}
	\centering{}\includegraphics[width=0.99\columnwidth]{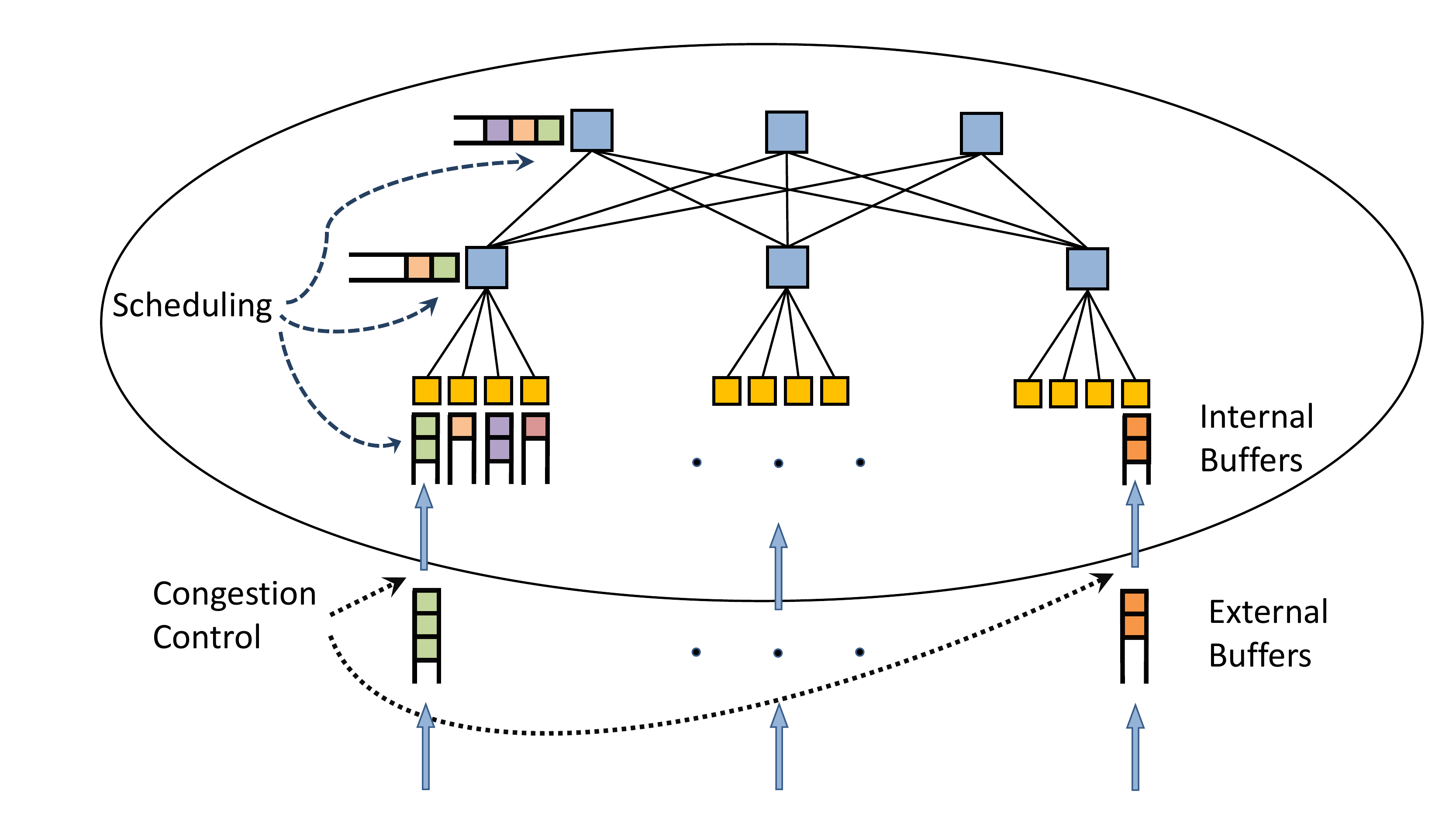}\caption{Overview of the congestion control and scheduling scheme.\label{fig:Overview}}
\end{figure}

Figure \ref{fig:congestion_control} illustrates our congestion control policy. The key idea
is to view the system as a bandwidth sharing network with flow-level
capacity constraints. The rate allocated to each flow buffered at source nodes is determined by an online algorithm that only uses
the queue lengths of the external buffers, and satisfies the capacity constraints. Another key
ingredient of our algorithm is a mechanism that translates the allocated rates to congestion control decisions. In particular,
we implement the congestion control algorithm at the granularity of flows, as opposed
to adjusting the rates on a packet-by-packet basis as in classical
TCP.

We consider a specific bandwidth allocation scheme called the store-and-forward
algorithm (SFA), which was first considered by Massouli\'{e} and later
discussed in \cite{bonald2003insensitive,kelly2009resource,proutiere_thesis,walton2009fairness}.
The SFA policy has been shown to be insensitive with
respect to general service time distributions \cite{zachary2007insensitivity},
and result in a reversible network with Poisson arrival processes~\cite{walton2009fairness}. The bandwidth sharing
network under SFA has a product-form queue size distribution in equilibrium.
Given this precise description of the stationary distribution, we
can obtain an explicit bound on the number of flows waiting at the
source nodes, which has the desirable form of $\ensuremath{O}(\#\text{hops}\ensuremath{\times}\text{flow-size}\ensuremath{/}\text{gap-to-capacity})$. 
Details of the congestion control policy description and analysis
are given in Section \ref{sec:congestion_control} to follow. 

\begin{figure}
	\begin{centering}
		\includegraphics[width=0.9\columnwidth]{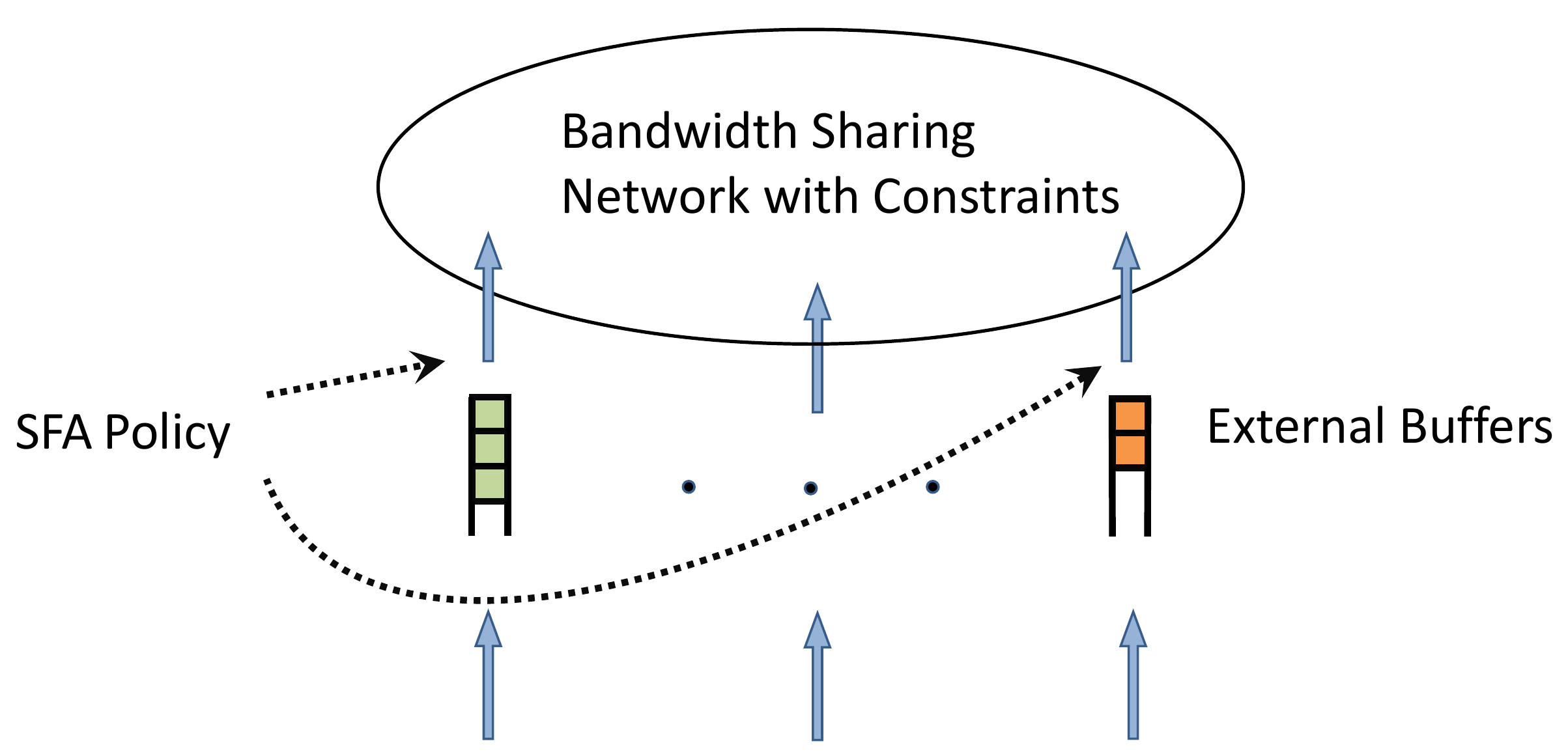}
		\par\end{centering}
	\caption{A congestion control policy based on the SFA algorithm for a bandwidth
		sharing network.\label{fig:congestion_control}}
	
\end{figure}

We also make use of the concept of emulation to design a packet scheduling
algorithm in the internal network, which is operated in discrete time.
In particular, we propose and analyze a scheduling mechanism that
is able to emulate a continuous-time \emph{quasi-reversible} network,
which has a highly desirable queue-size scaling. 

Our design consists of three elements. First, we specify the granularity
of timeslot in a way that maintains the same throughput as in a network
without the discretization constraint. By appropriately choosing the granularity,
we are able to address a general setting where flows arriving on
each route can have arbitrary sizes, as opposed to prior work that assumed unit-size
flows. Second, we consider a continuous-time network
operated under the Last-Come-First-Serve Preemptive-Resume (LCFS-PR)
policy. If flows on each route are assumed to arrive according a Poisson
process, the resulting queueing network is quasi-reversible with a
product-form stationary distribution. In this continuous-time setting,
we will show that the network achieves a flow delay bound of $\ensuremath{O}(\#\text{hops }\ensuremath{\times}\text{ flow-size }\ensuremath{/}\text{ gap-to-capacity})$.
Finally, we design a feasible scheduling policy for the discrete-time
network, which achieves the same throughput and delay bounds as the
continuous-time network. The resulting scheduling scheme is illustrated
in Figure \ref{fig:scheduling}.

\begin{figure}
	\begin{centering}
		\includegraphics[width=0.7\columnwidth]{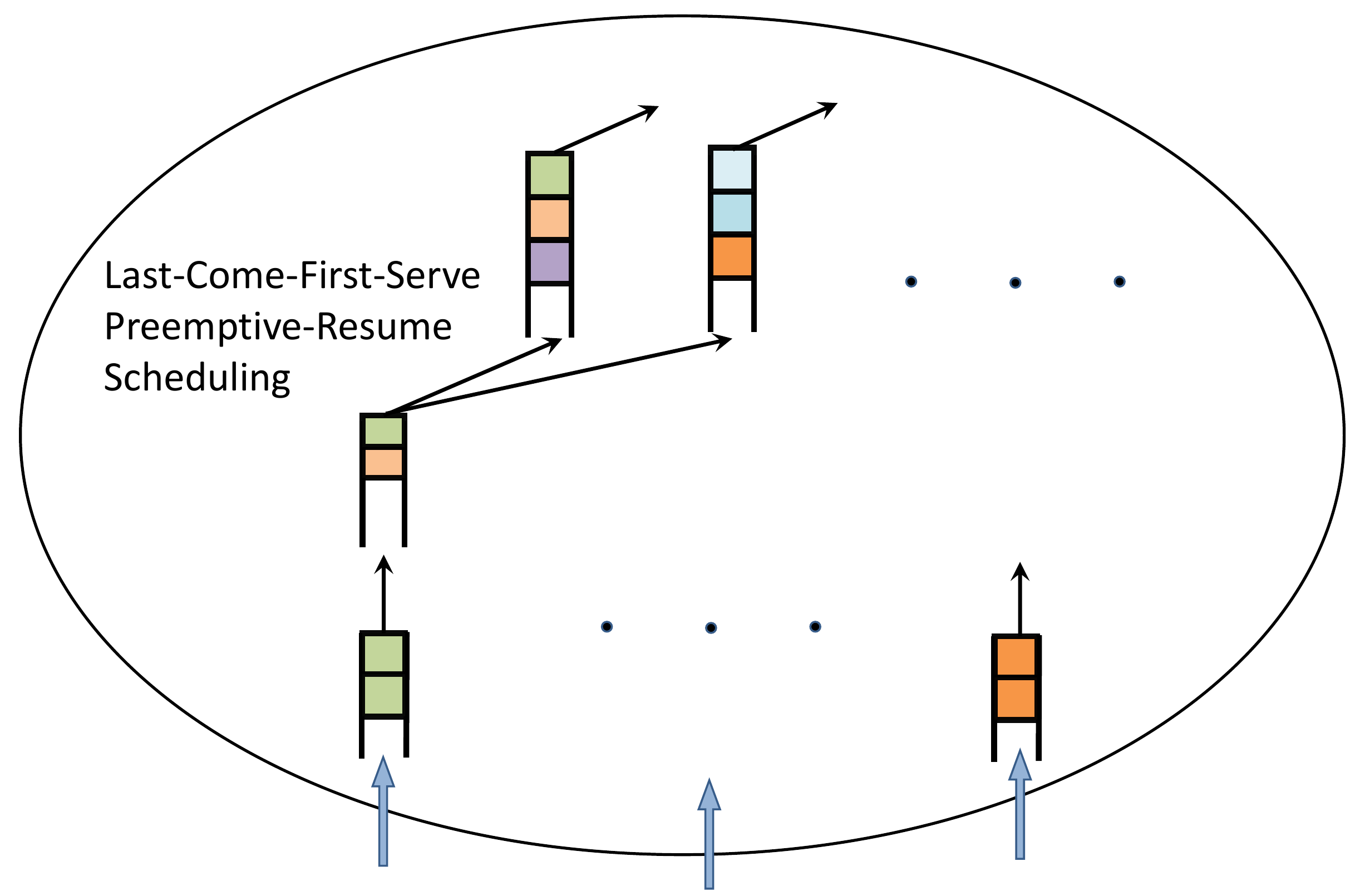}
		\par\end{centering}
	\caption{An adapted LCFS-PR scheduling algorithm \label{fig:scheduling}}
\end{figure}

\subsection{Our Contributions}

The main contribution of the paper is a centralized policy for both congestion control
and scheduling that achieves a per-flow end-to-end delay bound $\ensuremath{O}(\#\text{hops }\ensuremath{\times}$
$\text{flow-size }\ensuremath{/}\text{ gap-to-capacity})$. Some
salient aspects of our result are:
\begin{enumerate}
	\item The policy addresses both the congestion control and scheduling problems,
	in contrast to other previous work that focused on either congestion
	control or scheduling. 
	\item We consider flows with variable sizes.
	\item We provide per-flow delay bound rather than an aggregate bound. 
	\item Our results are non-asymptotic, in the sense that they hold for any admissible load.
	\item A central component of our design is the emulation of continuous-time
	quasi-reversible networks with a product-form stationary distribution.
	By emulating these queueing networks, we are able to translate the
	results therein to the network with congestion and scheduling constraints. 
	This emulation result can be of interest in its own right.
\end{enumerate}

\subsection{Organization}The remaining sections of the paper are organized
as follows. In Section \ref{sec:Model-and-Notation} we describe the
network model. The main results of the paper are presented in Section
\ref{sec:results}. The congestion control algorithm is described
and analyzed in Section \ref{sec:congestion_control}. Section \ref{sec:scheduling}
details the scheduling algorithm and its performance properties. We
discuss implementation issues and conclude the paper in Section \ref{sec:discussion}.

\section{Model and Notation\label{sec:Model-and-Notation}}

In this section, we introduce our model. We consider multi-rooted
tree topologies that have been widely adopted for datacenter networks~\cite{greenberg2009VL2_Clos,al2008scalable_fattree,singh2015Clos_Google}.
Such tree topologies can be represented by a directed
acyclic graph (DAG) model, to be described in detail shortly. We then describe
the network dynamics based on the graph model. At a higher level,
exogenous flows of various types arrive according to a continuous-time
process. Each flow is associated with a route, and has some amount
of data to be transferred along its route. As the internal
network is operated in discrete time, i.e., time is slotted for unit-size
packet transmission, each flow is broken into a number of packets.
Packets of exogenous flows would be buffered at the 
source nodes upon arrival, and then injected into the internal network according a specified
congestion control algorithm --- an example being the classical TCP. 
The packets traverse along nodes on their respective routes, queueing in 
buffers at intermediate nodes. It is worth noting that a flow departs the 
network once all of its packets reach the destination.

\subsection{An Acyclic Model}

Before providing the construction of a directed acyclic graph model
for the datacenter networks with tree topologies, we first introduce
the following definition of dominance relation between two queues. 
\begin{defn}[Dominance relation]
	For any two queues $Q_{1}$ and $Q_{2},$ $Q_{1}$ is said to be dominated
	by $Q_{2},$ denoted by $Q_{1}\dominate Q_{2},$ if there exists a flow
	for $Q_{2}$ such that its packets pass through $Q_{1}$ before entering
	$Q_{2}.$ 
\end{defn}
For a datacenter with a tree topology, each node $v$ maintains two
separate queues: a queue  $Q_{v}^{u}$ that sends data ``up'' the tree,
and another queue  $Q_{v}^{d}$ that sends data ``down'' the tree. Here the notions
of "up" and "down" are defined with respect to a fixed, chosen root of
the tree --- packets that are going towards nodes closer to the root 
are part of the ``up'' queue, while others are part of the ``down'' queue.

Consider the queues at leaf nodes of the tree. By definition, the up queues
cannot dominate any other queues, and the down queues
at leaf nodes cannot be dominated by any other queues. Therefore we can construct an order starting
with the up queues at leaf nodes $\{v_{1},\ldots,v_{k}\}$, i.e.,
$Q_{v_{1}}^{u},\ldots,Q_{v_{k}}^{u}.$ The down queues at leaf nodes
$Q_{v_{1}}^{d},\ldots,Q_{v_{k}}^{d}$ are placed at the end of the
order. Now we can remove the leaf nodes and consider the remaining
sub-tree. Inductively, we can place the up queues of layer-$l$ nodes
right after the up queues of layer-$(l+1)$ nodes, and place the down
queues of layer-$l$ nodes before those of layer-$(l+1)$ nodes. It is
easy to verify that the resulting order gives a partial order of queues
with respect to the dominance relation. Therefore, we can construct
a directed acyclic graph for the tree topology, where
each up/down queue $q_i$ is abstracted as a node. There exists a directed
link from $q_{i}$ to $q_{j}$ if there exists a flow such that its
packets are routed to $q_{j}$ right after departing $q_{i}.$

In the following, we represent the datacenter fabric by a directed acyclic
graph $G(\mathcal{V},\mathcal{E}),$ where $\mathcal{V}$ is the node
set, and $\mathcal{E} \subset \mathcal{V} \times \mathcal{V}$ is the set of directed links connecting nodes in
$\mathcal{V}$. The graph is assumed to be connected. Let $N := |\mathcal{V}|$ and
$M := |\mathcal{E}|$. Each entering flow is routed through the network
from its source node to its destination node. There are $J$ classes
of traffic routes indexed by $\mathcal{J}=\{1,2,\cdots,J\}$. Let
notation $v\in j$ or $j\in v$ denote that route~$j$ passes through
node $v.$ We denote by $\R\in\Real^{M\times J}$ the routing matrix
of the network, which represents the connectivity of nodes in different
routes. That is, 
\[
R_{mj}=\begin{cases}
1, & \text{if route \ensuremath{j} traverses link \ensuremath{m}}\\
0, & \text{otherwise}.
\end{cases}
\]

\subsection{Network Dynamics}

Each flow is characterized by the volume of data to be transferred
along its route, referred to as flow size. Note that the size of a
flow remains constant when it passes through the network. The \emph{flow size set} $\mathcal{X},$ which consists of all the sizes of flows arriving
at the system, is assumed to be countable.
Each flow can be classified according to its associated route and
size, which defines the \emph{type} of this flow. Let $\mathcal{T}:=\{(j,x):\;j\in\mathcal{J},\;x\in\mathcal{X}\}$
denote the set of flow types.

\medskip
\noindent\textbf{Flow arrivals. }External flow arrivals occur according to
a continuous-time stochastic process. The inter-arrival times of
type-$(j,x)$ flows are assumed to be i.i.d. with mean $\frac{1}{\lambda_{j,x}}$
(i.e., with rate $\lambda_{j,x}$) and finite variance. Let $\boldsymbol{\lambda}=(\lambda_{j,x}:\;j\in\mathcal{J},\;x\in\mathcal{X})$
denote the arrival rate vector (of dimension $J|\mathcal{X}|$). We denote the traffic intensity of route $j$
by the notation $\alpha_{j}=\sum_{x\in\mathcal{X}}\;x\lambda_{j,x}.$
Let 
\[
f_{v}=\sum_{j:\;j\in v}\sum_{x\in\mathcal{X}}\;x\lambda_{j,x}
\]
be the average service requirement for flows arriving at node $v$
per unit time. We define the effective load $\rho_{j}(\boldsymbol{\lambda})$
along the route $j$ as 
\[
\rho_{j}(\boldsymbol{\lambda})=\max_{v\in j}\;f_{v}.
\]

As mentioned before, exogenous flows are queued at the 
external buffers upon arrival. The departure process
of flows from the external buffer is determined by the congestion control
policy. The flows are considered to have departed from the external buffer (not necessarily from the network though) as soon
as the congestion control policy {\em fully accepts} it in the internal network for
transmission. 

\medskip
\noindent\textbf{Discretization.} The internal network is operated in discrete
time. Each flow is packetized. Packets of
injected flows become available for transmission only at integral
times, for instance, at the beginning of the timeslots. For a system
with timeslot of length $\epsilon,$ a flow of size $x$ is decomposed
into $\lceil\frac{x}{\epsilon}\rceil$ packets of equal size $\epsilon$
(the last packet may be of size less than a full timeslot). In a unit
timeslot each node is allowed to transmit at most one $\epsilon$-size
packet. If a node has less than a full timeslot worth of data to send,
unused network bandwidth is left wasted. In order to eliminate bandwidth
wastage and avoid throughput loss, the granularity of timeslot should
be chosen appropriately (cf.\ Section \ref{sec:scheduling}). 

Note that over a time interval $[k\epsilon,(k+1)\epsilon)$,
the transmission of packets at nodes occurs instantly at time $k\epsilon,$
while the arrival of new flows to the network occurs continuously
throughout the entire time interval.

\subsubsection{Admissible Region without Discretization}

By considering the network operating in continuous
time without congestion control, we can characterize the admissible region of the network without
discretization. Note that the service rate of each node is assumed
to be deterministically equal to 1. We denote by $\Q(t)=(Q_{1}(t),Q_{2}(t),\ldots,Q_{N}(t))$
the queue-size vector at time $t\in[0,\infty),$ where $Q_{i}(t)$
is the number of flows at node $i.$ Note that the process $\Q(t)$
alone is not Markovian, since each flow size is deterministic and the interarrival times have a general distribution. We
can construct a Markov description of the system, given by a process
$\Y(t),$ by augmenting the queue-size vector $\Q(t)$ with the remaining
size of the set of flows on each node, and the residual arrival times of flows of different types. The network system is said
to be stable under a policy if the Markov process $\Y(t)$ is positive
recurrent. An arrival rate vector is called (strictly) \emph{admissible }if and only
if the network with this arrival rate vector is stable under some policy. 
The \emph{admissible region} of the network is defined as the set of all admissible arrival rate vectors. 

Define a set ${\Lambda}\subset\Real_+^{J|\mathcal{X}|}$ of arrival rate vectors as follows:
\begin{equation}
{\Lambda} = \left\{ \blambda\in\Real_{+}^{J|\mathcal{X}|}:\;f_{v}=\sum_{j:\;j\in v}\sum_{x\in\mathcal{X}}\;x\lambda_{j,x} < 1,\forall v\in\mathcal{V}\right\} .\label{eq:capacity}
\end{equation}
It is easy to see that if $\blambda\notin {\Lambda},$ then the network is not stable under any scheduling policy, because in this case there exists at least one queue that receives load beyond its capacity. Therefore, the set ${\Lambda}$ contains the admissible region. As we will see in the subsequent sections, ${\Lambda}$ is in fact equal to the admissible region. In particular, for each $\blambda$ in $ \Lambda $, our proposed policy stabilizes the network with congestion control and discretization constraints. 

\subsection{Notation}

The end-to-end delay of a flow refers to the time interval between
the flow's arrival and the moment that all of its packets
reach the destination. We use $D^{(j,x),i}(\blambda)$ to denote
the end-to-end delay of the $i$-th flow of size $x$ along route
$j,$ for $j=1,\ldots,J,$ and $x\in\mathcal{X},$ under the joint
congestion control and scheduling scheme, and the arrival rate vector $\blambda.$
Note that $D^{(j,x),i}$ includes two parts: the time that a flow
spends on waiting in the external buffer, denoted by $D_{W}^{(j,x),i},$ and the
delay it experiences in the internal network, denoted by $D_{S}^{(j,x),i}.$ 
Then the sample mean of the waiting delay and that of the scheduling
delay over all type-$(j,x)$ flows are:
\begin{align*}
D_{W}^{(j,x)}(\blambda) & =\limsup_{k\rightarrow\infty}\frac{1}{k}\sum_{i=1}^{k}D_{W}^{(j,x),i}(\blambda),\\
D_{S}^{(j,x)}(\blambda) & =\limsup_{k\rightarrow\infty}\frac{1}{k}\sum_{i=1}^{k}D_{S}^{(j,x),i}(\blambda).
\end{align*}
When the steady state distributions exist, $D_{W}^{(j,x)}(\blambda)$
and $D_{S}^{(j,x)}(\blambda)$ are exactly the expected waiting delay
and scheduling delay of a type-$(j,x)$ respectively.
The expected total delay of a type-$(j,x)$ flow is: 
\[
D^{(j,x)}(\blambda)=D_{W}^{(j,x)}(\blambda)+D_{S}^{(j,x)}(\blambda).
\]

\section{Main Results\label{sec:results}}

In this section, we state the main results of our paper. 
\begin{thm}
\label{thm:delay-bound} Consider a network described in Section \ref{sec:Model-and-Notation}
with a Poisson arrival rate vector $\boldsymbol{\lambda}\in\Lambda$.
With an appropriate choice of time-slot granularity $\epsilon>0,$ 
there exists a joint congestion control and scheduling scheme for the 
discrete-time network such that for each flow type $(j, x)\in\mathcal{T}$,  
\begin{align*}
D^{(j,x),\epsilon} (\blambda)
& \leq C \left( \frac{xd_{j}}{1-\rho_{j}(\boldsymbol{\lambda})} + d_{j} \right)
\quad \mbox{with probability} ~1, 
\end{align*}
where $C > 0$ is a universal constant.
\end{thm}
The result stated in Theorem \ref{thm:delay-bound} holds for a network
with an arbitrary acyclic topology. In terms of the delay bound, it
is worth noting that a factor of the form $\frac{x}{1-\rho_{j}(\blambda)}$
is inevitable; in fact, the $\frac{1}{1-\rho}$ scaling behavior has
been universally observed in a large class of queueing systems. For
instance, consider a work-conserving $M/D/1$ queue where unit-size
jobs or packets arrive as a Poisson process with rate $\rho < 1$
and the server has unit capacity. Then, both average delay and queue-size
scale as $1/(1-\rho)$ as $\rho \to 1$. 

Additionally, the bounded delay implies that the system can be ``stabilized''
by properly choosing the timeslot granularity. In particular, our proposed algorithm 
achieves a desirable delay bound of $$O(\#\text{hops }\times \text{ flow-size }/ \text{ gap-to-capacity})$$
for any $\blambda$ that is admissible. That is, system is throughput optimal as well.

Theorem \ref{thm:delay-bound} assumes that flows arrive as Poisson
processes. The results can be extended to more general arrival processes
with i.i.d. interarrival times with finite means, by applying the
procedure of \emph{Poissonization} considered in \cite{gamal2006throughput_delay,jagabathula2008delay_scheduling}.
This will lead to similar bounds as stated in Theorem \ref{thm:delay-bound} but with a (larger) constant $C$ 
that will also depend on the distributional characterization of the arrival process; 
see Section~\ref{sec:discussion} for details.
\begin{rem}
Consider the scenario that the size of flows on route $j$ has a distribution
equal to a random variable $X_{i}.$ Then Theorem \ref{thm:delay-bound}
implies that the expected delay for flows along route $j$ can be
bounded as $D^{j,\epsilon}(\blambda)\leq C\left(\frac{\E[X_{j}]d_{j}}{1-\rho_{j}(\blambda)}+d_j \right).$
\end{rem}

In order to characterize the network performance under the proposed joint scheme, we establish delay bounds achieved by the proposed congestion control policy and the scheduling algorithm, separately. The following theorems summarize the performance properties of our approach. 

The first theorem argues that by using an appropriate congestion control policy, the  delay of flows due to waiting at source nodes before
getting ready for scheduling (induced by congestion control) is well behaved. 

\begin{thm}
\label{thm:congestion} Consider a network described in Section \ref{sec:Model-and-Notation}
with a Poisson arrival rate vector $\boldsymbol{\lambda}\in\Lambda$,
operating under the SFA congestion control policy described in Section
\ref{sec:congestion_control}. Then, for each flow type $(j,x)\in\mathcal{T}$,
the per-flow waiting delay of type-$ (j,x) $ flows is bounded as
\begin{align*}
D_{W}^{(j,x)} (\blambda)& \leq\frac{xd_{j}}{1-\rho_{j}(\boldsymbol{\lambda})}, \quad \text{with probability}~1.
\end{align*}
\end{thm}

The second result is about scheduling variable length packets in the internal network in discrete time. If the network could
operate in continuous time, arrival processes were Poisson and packet scheduling in the internal network were Last-Come-First-Serve Preemptive
Resume (LCFS-PR), then the resulting delay would be the desired quantity due to the classical result about product-form stationary
distribution of quasi-reversible queueing networks, cf. \cite{bcmp1975network, kelly1979reversibility}. However, in our setting, since network operates in discrete time, we cannot use LCFS-PR (or for that matter any of the quasi-reversible policies). As the main contribution of this paper, we show that there is a
way to {\em adapt} LCFS-PR to the discrete time setting so that for each sample-path, the departure time of each flow is bounded above by its departure time from the ideal continuous-time network. Such a powerful emulation leads to Theorem \ref{thm:scheduling}
stated below. We note that \cite{jagabathula2008delay_scheduling} provided such an emulation scheme when all flows had the same (unit) packet size. Here we provide a generalization
of such a scheme for the setting where variable flow sizes are present.\footnote{We note that the result stated in \cite{jagabathula2008delay_scheduling} requires
that the queuing network be acyclic.}
\begin{thm}
\label{thm:scheduling} Consider a network described in Section \ref{sec:Model-and-Notation}
with a Poisson arrival rate vector $\boldsymbol{\lambda}\in\Lambda$, operating
with an appropriate granularity of timeslot $\epsilon>0,$ under the
adapted LCFS-PR scheduling policy to be described in Section \ref{sec:scheduling}.
Then, for each $ (j,x) \in \mathcal{T} $, 
the per-flow scheduling delay of type-$(j,x)$ flows is bounded as 
\begin{align*}
D_{S}^{(j,x),\epsilon}(\blambda) 
& \leq C' \left( \frac{xd_{j}}{1-\rho_{j}(\boldsymbol{\lambda})} + d_{j} \right)
 \quad \mbox{with probability}~1,
\end{align*}
where $C'>0$ is a universal constant. 
\end{thm}

As discussed, for Theorem~\ref{thm:scheduling} to be useful, it is essential to have the flow departure process from the first stage (i.e., the congestion 
control step) be Poisson. To that end, we state the following Lemma, which follows from an application of 
known results on the reversibility property of the bandwidth-sharing network under the SFA policy~\cite{walton2009fairness}.
 \begin{lem} \label{lem:poisson}
 	Consider the network described in Theorem~\ref{thm:congestion}. Then for each $(j,x)\in\mathcal{T},$ the departure times of type-$(j,x)$ flows form an independent Poisson process of rate $\lambda_{j,x}.$
 \end{lem}

Therefore, Theorem~\ref{thm:delay-bound} is an immediate consequence of Theorems~\ref{thm:congestion}--\ref{thm:scheduling} and Lemma~\ref{lem:poisson}.

\section{Congestion Control: An Adapted SFA Policy \label{sec:congestion_control}}

In this section, we describe an online congestion control policy,
and analyze its performance in terms of explicit bounds on the flow
waiting time, as stated in Theorem \ref{thm:congestion} in Section
\ref{sec:results}. We start by describing the 
store-and-forward allocation policy (SFA) that is introduced in \cite{proutiere_thesis}
and analyzed in generality in \cite{bonald2003insensitive}. We describe
its performance in Section \ref{subsec:SFA}. 
Section \ref{subsec:congestion_control} details our congestion control 
policy, which is an adaptation of the SFA policy.

\subsection{Store-and-Forward Allocation (SFA) Policy\label{subsec:SFA}}

We consider a bandwidth-sharing networking operating in continuous
time. A specific allocation policy of interest is the store-and-forward
allocation (SFA) policy. 

\medskip
\noindent\textbf{Model. }Consider a continuous-time network with a set $\text{\ensuremath{\mathcal{L}}}$
of resources. Suppose that each resource $l\in\mathcal{\mathcal{L}}$
has a capacity $C_{l}>0.$ Let $\mathcal{J}$ be the set of routes.
Suppose that $|\mathcal{L}|=L$ and $|\mathcal{J}|=J$.
Assume that a unit volume of flow on route $j$ consumes an amount
$B_{lj}\geq0$ of resource $l$ for each $l\in\mathcal{L}, $ where
$B_{lj}>0$ if $l\in j$, and $B_{lj}=0$ otherwise. The simplest
case is where $B_{lj}\in\{0,1\},$ i.e., the matrix $B=(B_{lj},\;l\in\mathcal{L},\;j\in\mathcal{J})\in\Real_{+}^{L\times J}$
is a $0$-$1$ incidence matrix. We denote by $\mathcal{\mathcal{K}}$
the set of all resource-route pairs $(l,j)$ such that route $j$
uses resource $l$, i.e., $\mathcal{K}=\left\{ (l,j):\;B_{lj}>0\right\} $;
assume that $|\mathcal{K}|=K.$

Flows arrive to route $j$ as a Poisson process of rate $\lambda_i$. For a flow of size
$x$ on route~$j$ arriving at time $t_{\text{start}},$ its completion
time $t_{\text{end}}$ satisfies 
\[
x=\int_{t_{\text{start}}}^{t_{\text{end}}}c(t)dt,
\]
where $c(t)$ is the bandwidth allocated to this flow on each link
of route $j$ at time $t.$ The size of each flow on route $j$ is
independent with distribution equal to a random variable $X_{j}$
with finite mean. Let $\mu_{j}=(\E[X_{j}])^{-1},$ and $\alpha_{j}=\frac{\lambda_{j}}{\mu_{j}}$
be the traffic intensity on route $j.$

Let $N_{j}(t)$ record the number of flows on route $j$ at time $t,$
and let $\N(t)=(N_{1}(t),N_{2}(t),\ldots,N_{J}(t)).$ Since the service
requirement follows a general distribution, the queue length process
$\N(t)$ itself is not Markovian. Nevertheless, by augmenting the
queue-size vector $\N(t)$ with the residual workloads of the set
of flows on each route, we can construct a Markov description of the
system, denoted by a process $\Z(t).$ On average, $\alpha_{j}$
units of work arrive at route~$j$ per unit time. Therefore, a necessary
condition for the Markov process $\Z(t)$ to be positive recurrent
is that 
\[
\sum_{j:\;l\in j}B_{lj}\alpha_{j}<C_{l},\quad\forall l\in\mathcal{L}.
\]
An arrival rate vector $\blambda=(\lambda_{j}:j\in\mathcal{J})$ is
said to be admissible if it satisfies the above condition. Given
an admissible arrival rate vector $\blambda,$ we can define the load
on resource $l$ as 
\[
g_{l}(\blambda)=\left(\sum_{j:\;j\in l}B_{lj}\alpha_{j}\right)/C_{l}.
\]

We focus on allocation polices such that the total bandwidth allocated
to each route only depends on the current queue-size vector $\n=(n_{j}:\;j\in\mathcal{J})\in\Int_{+}^{J},$ which represents
the number of flows on each route. We consider
a processor-sharing policy, i.e., the total bandwidth $\phi_{j}$
allocated to route $j$ is equally shared between all $n_{j}$ flows.
The bandwidth vector $\mathbf{\bphi}(\n)=(\phi_{1}(\n),\ldots,\phi_{J}(\n))$
must satisfy the capacity constraints
\begin{equation}
\sum_{j:\;j\in l}B_{lj}\phi_{j}(\n)\leq C_{l},\qquad\forall l\in\mathcal{L}.\label{eq:capacity_constraint}
\end{equation}

\smallskip
\noindent\textbf{Store-and-Forward Allocation (SFA) policy. }The SFA policy
was first considered by Massouli\'{e} (see page 63, Section 3.4.1 in~\cite{proutiere_thesis}) and
later analyzed in the thesis of Prouti\`{e}re \cite{proutiere_thesis}.
It was shown by Bonald and Prouti\`{e}re \cite{bonald2003insensitive}
that the stationary distribution induced by this policy has a product
form and is insensitive for phase-type service distributions. Later
Zachary established its insensitivity for general service distributions
\cite{zachary2007insensitivity}. Walton \cite{walton2009fairness}
and Kelly et al.\ \cite{kelly2009resource} discussed the relation
between the SFA policy, the proportionally fair allocation and multiclass
queueing networks. Due to the insensitivity property of SFA, the invariant
measure of the process $\N(t)$ only depends on the parameters $\blambda$
and $\boldsymbol{\mu}.$ 

We introduce some additional notation to describe this policy. Given
the vector $\n=(n_{j}:\;j\in\mathcal{J})\in\Int_{+}^{J},$ which represents
the number of flows on each route, we define the set 
\begin{align*}
U(\n) & =\left\{ \m=(m_{lj}:\;(l,j)\in\mathcal{K})\in\Int_{+}^{K}:\;n_{j}=\sum_{l:\;l\in j}m_{lj},\;\forall j\in\mathcal{J}\right\} .
\end{align*}
With a slight abuse of notation, for each $\m\in\Int_{+}^{K},$ we
define $m_{l}:=\sum_{j:\;l\in j}m_{lj}$ for all $l\in\mathcal{L}.$
We also define quantity
\[
{m_{l} \choose m_{lj}:\;j\ni l}=\frac{m_{l}!}{\prod_{j:\;l\in j}(m_{lj}!)}.
\]
 For $\n\in\Int_{+}^{J},$ let
\[
\Phi(\n)=\sum_{\m\in U(\n)}\prod_{l\in\mathcal{L}}\left({m_{l} \choose m_{lj}:\;j\ni l}\prod_{j:\;l\in j}\left(\frac{B_{lj}}{C_{l}}\right)^{m_{lj}}\right).
\]
 We set $\Phi(\n)=0$ if at least one of the components of $\n$ is
negative. 

The SFA policy assigns rates according to the function $\bphi:\Int_{+}^{J}\rightarrow\Real_{+}^{J},$
such that for any $\n\in\Int_{+}^{J},$ $\bphi(\n)=(\phi_{j}(\n))_{j=1}^{J},$
with 
\[
\phi_{j}(\n) = \frac{\Phi(\n-\e_{j})}{\Phi(\n)},
\]
 where $\e_{j}$ is the $j$-th unit vector in $\Int_{+}^{J}.$

\smallskip

The SFA policy described above has been shown to be feasible, i.e.,
$\bphi$ satisfies condition (\ref{eq:capacity_constraint}) \cite[Corollary 2]{kelly2009resource},
\cite[Lemma 4.1]{walton2009fairness}. Moreover, prior work has established
that the bandwidth-sharing network operating under the SFA policy
has a product-form invariant measure for the number of waiting
flows \cite{bonald2003insensitive,walton2009fairness,kelly2009resource,zachary2007insensitivity}, and the measure is insensitive to the flow size distributions~\cite{zachary2007insensitivity,walton2009fairness}. The above work is summarized in the following theorem~\cite[Theorem 4.1]{shah2014SFA}.
\begin{thm}
\label{thm:SFA} Consider a bandwidth-sharing network operating under
the SFA policy described above. If
\[
\sum_{j:\;l\in j}B_{lj}\alpha_{j}<C_{l},\quad\forall l\in\mathcal{L},
\]
then the Markov process $\Z(t)$ is positive
recurrent, and $\N(t)$ has a unique stationary distribution $\bpi$ given by
\[
\bpi(\n)=\frac{\Phi(\n)}{\Phi}\prod_{j\in\mathcal{J}}\alpha_{j}^{n_{j}},\;\n\in\Int_{+}^{J}, \label{eq:SFA_distribution}
\]
 where 
\[
\Phi=\prod_{l\in\mathcal{L}}\left(\frac{C_{l}}{C_{l}-\sum_{j:\;l\in j}B_{lj}\alpha_{j}}\right).
\]

\end{thm}
By using Theorem \ref{thm:SFA} (also see \cite[Propositions 4.2 and 4.3]{shah2014SFA}),
an explicit expression can be obtained for $\E[N_{j}],$  the expected number of flows on each route.
\begin{prop}
\label{prop:delay_BN}Consider a bandwidth-sharing network operating
under the SFA policy, with the arrival rate vector $\blambda$ satisfying
$$\sum_{j:\;l\in j}B_{lj}\alpha_j<C_{l},\;\forall l\in\mathcal{L}.$$
For each $l\in\mathcal{L},$ let $g_{l}=\left(\sum_{j:\;l\in j}B_{lj}\alpha_j\right)/C_{l}.$
Then $\N(t)$ has a unique stationary distribution. In particular,
for each $j\in\mathcal{J},$ 
\[
\E[N_{j}]=\sum_{l:\;l\in j}\frac{B_{lj}\alpha_j}{C_{l}}\frac{1}{1-g_{l}}.
\]
\end{prop}

\begin{proof}
	Define a measure $\tilde{\bpi}$ on $\Int_{+}^{K}$ as follows \cite[ Proposition 4.2 ]{shah2014SFA}:
	for each $\m\in\Int_{+}^{K},$
	\[
	\tilde{\bpi}(\m)=\frac{1}{\Phi}\prod_{l=1}^{L}\left({m_{l} \choose m_{lj}:\;j\ni l}\prod_{j:\;l\in j}\left(\frac{B_{lj}\alpha_{j}}{C_{l}}\right)^{m_{lj}}\right).
	\]
	It has been shown that $\tilde{\bpi}$ is the stationary distribution
	for a multi-class network with processor sharing queues \cite[Proposition 1]{kelly2009resource}, \cite[Proposition 2.1]{walton2009fairness}.
	Note that the marginal distribution of each queue in equilibrium is
	the same as if the queue were operating in isolation. To be precise,
	for each queue $l,$ the queue size is distributed as if the queue has Poisson
	arrival process of rate $g_{l}.$ Hence the size of queue $l,$ denoted
	by $M_{l},$ satisfies $\E_{\tilde{\bpi}}[M_{l}]=\frac{g_{l}}{1-g_{l}}.$ Let $M_{lj}$
	be the number of class $j$ customers at queue $l.$ Then by Theorem
	3.10 in \cite{kelly1979reversibility}, we have
	\begin{equation}
	\E_{\tilde{\bpi}}[M_{lj}]=\frac{B_{lj}\alpha_{j}}{\sum_{i:\;l\in i}B_{li}\alpha_{i}}\E_{\tilde{\bpi}}[M_{l}]=\frac{B_{lj}\alpha_{j}}{C_{l}}\frac{1}{1-g_{l}}.\label{eq:Exp_Mlj}
	\end{equation}

	Given the above expressions of $\tilde{\bpi}$ and $\bpi$, we can show that $\bpi(\n)=\sum_{\m\in U(\n)}\tilde{\bpi}(\m),$ $\forall\n\in\Int_{+}^{J}$.
	This fact has been shown by Bonald and Prouti\`{e}re \cite{bonald2004performance}
	and Walton \cite{walton2009fairness}. Hence we have 
	\begin{align*}
	\E_{\bpi}[N_{j}] & =\sum_{k=0}^{\infty}k\left(\sum_{\n\in\Int_{+}^{J}:n_{j}=k}\bpi(\n)\right)=\sum_{k=0}^{\infty}k\left(\sum_{\n\in\Int_{+}^{J}:n_{j}=k}\sum_{\m\in U(\n)}\tilde{\bpi}(\m)\right)\\
	& =\sum_{k=0}^{\infty}k\left(\sum_{\m\in\Int_{+}^{K}}\Indicator(\sum_{l:\;l\in j}m_{lj}=k)\tilde{\bpi}(\m)\right)\\
	& =\E_{\tilde{\bpi}} \Big[ \sum_{l:\;l\in j}M_{lj} \Big] \\
	& =\sum_{l:\;l\in j}\frac{B_{lj}\alpha_{j}}{C_{l}}\frac{1}{1-g_{l}}
	\end{align*}
	as desired.
\end{proof}

We now consider the departure processes of the bandwidth-sharing network. It is a known result that the bandwidth sharing network with the SFA policy is reversible, as summarized in the proposition below. This fact has been observed by Prouti\`{e}re~\cite{proutiere_thesis}, 
Bonald and Prouti\`{e}re~\cite{bonald2004performance}, and Walton~\cite{walton2009fairness}. 
Note that Lemma~\ref{lem:poisson} follows immediately from the reversibility property of the network and the fact that flows of each type arrive according to a Poisson process.
\begin{prop}[Proposition 4.2 in~\cite{walton2009fairness}] \label{prop:reversibility} 
Consider a bandwidth sharing network operating under SFA. If $\sum_{j:\;l\in j}B_{lj}\alpha_{j}<C_{l},\;\forall l\in\mathcal{L},$ then the network is reversible.
\end{prop}

\subsection{An SFA-based Congestion Control Scheme\label{subsec:congestion_control}}

We now describe a congestion control scheme for a general acyclic
network. As mentioned earlier, every \emph{exogenous} flow is pushed to the corresponding external buffer upon arrival. The internal buffers store flows that are either transmitted from other nodes or
moved from the external buffers. We want to emphasize that
only packets present in the internal queues are eligible
for scheduling. The congestion control mechanism determines the time
at which each external flow should be admitted into the network for
transmission, i.e., when a flow is removed from the external buffer
and pushed into the respective source internal buffer.  

The key idea of our congestion control policy is as follows. Consider
the network model described in Section \ref{sec:Model-and-Notation},
denoted by $\mathcal{N}$. Let $A$ be an $N\times J$ matrix where
\[
A_{nj}=\begin{cases}
1, & \text{if route \ensuremath{j} passes through node \ensuremath{n}}\\
0, & \text{otherwise.}
\end{cases}
\]
The corresponding admissible region $\Lambda$ given by (\ref{eq:capacity})
can be equivalently written as 
\[
\Lambda=\left\{ \blambda\in\Real_{+}^{J|\mathcal{X}|}:\;A\balpha<\boldsymbol{1}\right\} ,
\]
 where \emph{$\balpha=(\alpha_{j},\;j\in\mathcal{J})$ }with\emph{
$\alpha_{j}:=\sum_{x\in\mathcal{X}}\;x\lambda_{j,x},$ }and $\boldsymbol{1}$
is an $N$-dim vector with all elements equal to one. 

Now consider a virtual bandwidth-sharing network, denoted by \textbf{$\mathcal{N}_{B}$},\textbf{
}with $J$ routes which have a one-to-one correspondence with the
$J$ routes in $\mathcal{N}$. The network $\mathcal{N}_{B}$ has
$N$ resources, each with unit capacity. The resource-route relation
is determined precisely by the matrix $A.$ Observe that each resource
in $\mathcal{N}_{B}$ corresponds to a node in $\mathcal{N}.$ Assume
that the exogenous arrival traffic of \textbf{$\mathcal{N}_{B}$ }is
identical to that of $\mathcal{N}$. That is, each flow arrives at
the same route of $\mathcal{N}$\textbf{ }and\textbf{ $\mathcal{N}_{B}$
}at the same time. \textbf{$\mathcal{N}_{B}$ }will operate under
the insensitive SFA policy described in Section \ref{subsec:SFA}. 

\smallskip

\noindent\textbf{Emulation scheme} $E$. Flows arriving at $\mathcal{N}$ will
be admitted into the network for scheduling in the following way.
For a flow $f$ arriving at its source node $v_{s}$ at time $t_{f}$,
let $\delta_{f}$ denote its departure time from \textbf{$\mathcal{N}_{B}$}.
In the network $\mathcal{N},$ this flow will be moved from the external buffer to the internal queue of node $v_{s}$ at time $\delta_{f}.$
Then all packets of this flow will simultaneously become eligible
for scheduling. Conceptually, all flows will be first fed into the
virtual network \textbf{$\mathcal{N}_{B}$ }before being admitted
to the internal network for transmission. 

\smallskip
This completes the description of the congestion control policy. Observe that the centralized controller simulates the virtual network $\mathcal{N}_{B}$ with the SFA policy in parallel, and tracks the departure processes of flows in $\mathcal{N}_{B}$ to make congestion control decisions for $\mathcal{N}$. According to the above emulation scheme $E$, we have the following lemma.
\begin{lem}
\label{lem:delay_flow_congestion}Let $D_{B}$ denote the delay of
a flow $f$ in $\mathcal{N}_{B},$ and $D_{W}$ be the amount of time
the flow spends at the external buffer in $\mathcal{N}$. Then
$D_{W}=D_{B}$ surely. 
\end{lem}

For each flow type $(j,x)\in\mathcal{T},$ let $D_{B}^{(j,x)} (\blambda)$ denote
the sample mean of the delays over all type-$(j,x)$ flows in the bandwidth
sharing network $\mathcal{N}_{B}.$ That is, $D_{B}^{(j,x)} (\blambda) := \limsup_{k\rightarrow\infty}\frac{1}{k}\sum_{i=1}^{k}D_{B}^{(j,x),i} (\blambda),$
where $D_{B}^{(j,x),i} (\blambda)$ is the delay of the $i$-th type-$(j,x)$
flow. From Theorem~\ref{thm:SFA} and Proposition~\ref{prop:delay_BN}, we readily deduce the following
result.

\begin{prop}
\label{prop:delay_SFA}In the network $\mathcal{\mathcal{N}}_{B},$
we have 

\[
D_{B}^{(j,x)} (\blambda)=\sum_{v:v\in j}\frac{x}{1-f_{v}},\quad\text{with probability \ensuremath{1.}}
\]
\end{prop}

\begin{proof}
 Recall that in the network $\mathcal{N}_{B},$ flows arriving on each route are
further classified by the size. Each type-$(j,x)$ flow
requests $x$ amount of service, deterministically. It is not difficult to see that properties of the bandwidth-sharing network stated in Section~\ref{subsec:SFA} still hold with the refined classification. We denote by $N_{j,x}(t)$
the number of type-$(j,x)$ flows at time $t.$ A Markovian description
of the system is given by a process $\Z(t)$, whose state contains the queue-size
vector $\N(t)=(N_{j,x}(t),j\in\mathcal{J},x\in\mathcal{X})$ and the residual workloads of the set of flows
on each route. 

By construction, each resource $v$ in the network $\mathcal{N}_{B}$ corresponds to
a node $v$ in the network $\mathcal{N},$ with unit capacity. The load $ g_v $ on the resource
$v$ of $ \mathcal{N}_{B} $ and the load $ f_v $ on node $ v $ of $ \mathcal{N} $ satisfy the relationship 
\[
g_{v}=\sum_{v:v\in j}A_{vj}\alpha_{j}=f_{v}.
\]
As $\blambda\in\Lambda,$ $f_{v}<1$ for each $v\in\mathcal{V}.$
By Theorem \ref{thm:SFA}, the Markov process $\Z(t)$ is positive
recurrent, and $\N(t)$ has a unique stationary distribution. Let $N_{j,x}$
denote the number of type-$(j,x)$ flow in the bandwidth sharing network
$\mathcal{N}_{B}$, in equilibrium. Following the same argument as
in the proof of Proposition \ref{prop:delay_BN}, we have 
\[
\E[N_{j,x}]=\sum_{v:\;v\in j}\frac{x\lambda_{j,x}}{1-f_{v}}.
\]

Consider the delay of a type-$(j,x)$ flow in equilibrium, denoted
by $\bar{D}_{B}^{(j,x)} (\blambda)$. By applying Little's Law to the stationary
system concerning only type-$(j,x)$ flows, we can obtain 
\[
\E[\bar{D}_{B}^{(j,x)} (\blambda)]=\frac{\E[N_{j,x}]}{\lambda_{j,x}}=\sum_{v:\;v\in j}\frac{x}{1-f_{v}}.
\]

By the ergodicity of the network $\mathcal{N}_{B},$ we have
\[
D_{B}^{(j,x) (\blambda)}=\E[\bar{D}_{B}^{(j,x)} (\blambda)],\quad\text{with probability}\;1,
\]
thereby completing the proof of Proposition~\ref{prop:delay_SFA}.
\end{proof}

Equipped with Proposition~\ref{prop:delay_SFA} and Lemma \ref{lem:delay_flow_congestion}, we are now ready to prove Theorem \ref{thm:congestion}.
\begin{proof}[Proof of Theorem \ref{thm:congestion}] 

By the definition, we have $\frac{1}{1-f_{v}}\leq\frac{1}{1-\rho_{j}(\blambda)}$,
for each $v\in j.$ Proposition \ref{prop:delay_SFA} implies that, with probability 1,
\[
D_{B}^{(j,x)} (\blambda)=\sum_{v:v\in j}\frac{x}{1-f_{v}}\leq\frac{xd_{j}}{1-\rho_{j}(\blambda)},
\]
It follows from Lemma \ref{lem:delay_flow_congestion} that $D_{W}^{(j,x)} (\blambda) = D_{B}^{(j,x)} (\blambda).$
Therefore, we have
\[
D_{W}^{(j,x)} (\blambda) \leq\frac{xd_{j}}{1-\rho_{j}(\blambda)},\quad\text{with probability}\;1.
\] 
This completes the proof of Theorem~\ref{thm:congestion}.
\end{proof}

\section{Scheduling: Emulating LCFS-PR\label{sec:scheduling}}

In this section, we design a scheduling algorithm that achieves a
delay bound of 
$$O\Big(\# \text{hops} \times  \text{flow-size} \times \frac{1}{\text{gap-to-capacity}}\Big)$$ 
 without throughput loss for a multi-class queueing network operating in discrete
 time. As mentioned earlier, our design consists of three key steps.
We first discuss how to choose the granularity of time-slot $\epsilon$
appropriately, in order to avoid throughput loss in the discrete-time
network. In terms of the delay bound, we leverage the results from
a continuous-time network operating under the Last-Come-First-Serve
Preemptive-Resume (LCFS-PR) policy. With a Poisson arrival process,
the network becomes quasi-reversible with a product-form stationary
distribution. Then we adapt the LCFS-PR policy to obtain a scheduling
scheme for the discrete-time network. In particular, we design an emulation
scheme such that the delay of a flow in the discrete-time network
is bounded above by the delay of the corresponding flow in the continuous-time
network operating under LCFS-PR. This establishes the delay property
of the discrete-time network.

\subsection{Granularity of Discretization }

We start with the following definitions that will be useful going forward.
\begin{defn}[$\mathcal{N}_{D}^{(\epsilon)}$ network]
It is a discrete-time network with the topology, service requirements and
exogenous arrivals as described in Section \ref{sec:Model-and-Notation}.
In particular, each time slot is of length $\epsilon$. Additionally,
the arrivals of size-$x$ flows on route $j$ form a Poisson process
of rate $\lambda_{j,x}.$ For such a type-$(j,x)$ flow, its size in 
$\mathcal{N}_{D}^{(\epsilon)}$
becomes
\[
x^{(\epsilon)}:=\epsilon\left\lceil \frac{x}{\epsilon}\right\rceil ,
\]
and the flow is decomposed into $\lceil\frac{x}{\epsilon}\rceil$
packets of equal size $\epsilon$. 
\end{defn}
\begin{defn}[$\mathcal{N}_{C}^{(\epsilon)}$  network]
It is a continuous-time network with the same topology, service requirements
and exogenous arrivals as $\mathcal{N}_{D}^{(\epsilon)}.$ The size
of a type-$(j,x)$ flow is modified in the same way as that in $\mathcal{N}_{D}^{(\epsilon)}$. 
\end{defn}

Consider the network $\mathcal{N}_{D}^{(\epsilon)}.$ Given an arrival
rate vector $\blambda\in\Lambda,$ the load on node $v$ is
\begin{align*}
f_{v}^{(\epsilon)}= & \sum_{j:j\in v}\sum_{x:x\in\mathcal{X}}\epsilon\left\lceil \frac{x}{\epsilon}\right\rceil \lambda_{j,x}.
\end{align*}
As each node has a unit capacity, $\blambda$ is admissible in $\mathcal{N}_{D}^{(\epsilon)}$
only if $f_{v}^{(\epsilon)}<1$ for all $v\in\mathcal{V}$. We let 
\begin{align}
\epsilon & = \min \left\{
	\frac{1}{C_{0}}\cdot\min_{v\in\mathcal{V}}\left\{ \frac{1-f_{v}}{\sum_{j:j\in v}\sum_{x:x\in\mathcal{X}}\lambda_{j,x}}\right\},
	\;
	\min_{j\in \mathcal{J}} \left\{ 1- \rho_j(\blambda) \right\} 
\right\},
\label{eq:epsilon}
\end{align}
where $C_{0}>1$ is an arbitrary constant. Then for each $v\in\mathcal{V},$
\begin{align*}
f_{v}^{(\epsilon)} 
& < \sum_{j:j\in v}\sum_{x:x\in\mathcal{X}}(x+\epsilon)\lambda_{j,x} \\
&=  f_{v} + \epsilon \sum_{j:j\in v}\sum_{x:x\in\mathcal{X}} \lambda_{j,x} \\
&\overset{(a)}{\leq}f_{v}+\frac{1-f_{v}}{C_{0}},\label{eq:node_load_Nd}
\end{align*}
where the inequality (a) holds because $\epsilon \le \frac{1-f_{v}}{C_{0}} \cdot \frac{1}{\sum_{j:j\in v}\sum_{x:x\in\mathcal{X}}\lambda_{j,x}} $ by the definition of $\epsilon$
in Eq.\ (\ref{eq:epsilon}). Since $\blambda\in\Lambda,$ for each
$v\in\mathcal{V},$ $f_{v}<1$. We thus have 
\begin{equation}
1-f_{v}^{(\epsilon)}\geq\frac{C_{0}-1}{C_{0}}(1-f_{v})>0.\label{eq:node_load_Nd_gap}
\end{equation}
Thus, each $\blambda\in\Lambda$ is admissible in the discrete-time
network $\mathcal{N}_{D}^{(\epsilon)}$. 

\subsection{Property of $\mathcal{N}_{C}^{(\epsilon)}$\label{subsec:Nc}}

Consider the continuous-time open network $\mathcal{N}_{C}^{(\epsilon)}$,
where a LCFS-PR scheduling policy is used at each node. From Theorems
$3.7$ and $3.8$ in~\cite{kelly1979reversibility}, the network $\mathcal{N}_{C}^{(\epsilon)}$ has a product-form queue length distribution in equilibrium, if the
following conditions are satisfied: 
\begin{enumerate}
\item[(C1.)] the service time distribution is either phase-type or the limit of
a sequence of phase-type distributions; 
\item[(C2.)] the total traffic at each node is less than its capacity. 
\end{enumerate}
Note that the sum of $n$ exponential random variables each with mean
$\frac{1}{nx}$ has a phase-type distribution and converges in distribution
to a constant $x,$ as $ n $ approaches infinity. Thus the first condition
is satisfied. For each $\blambda\in\Lambda,$ the second condition
holds for $\mathcal{N}_{C}^{(\epsilon)}$ with $\epsilon$ defined
as Eq.\ (\ref{eq:epsilon}).

In the following theorem, we establish a bound for the delay
experience by each flow type. For $j\in\mathcal{J},x\in\mathcal{X},$ let $D^{(j,x),\epsilon} (\blambda)$ denote the sample
mean of the delay over all type-$(j,x)$ flows, i.e.,
\begin{align*}
D^{(j,x),\epsilon} (\blambda) & =\limsup_{k\rightarrow\infty}\frac{1}{k}\sum_{i=1}^{k}D^{(j,x),\epsilon,i} (\blambda),
\end{align*}
where $D^{(j,x),\epsilon,i}$ is the delay of the $i$-th type-$(j,x)$ flow.

\begin{thm}
\label{thm:Delay_Nc} 
In the network $\mathcal{N}_{C}^{(\epsilon)},$ we have
\[
{D}^{(j,x),\epsilon} (\blambda) =\sum_{v:v\in j}\frac{\xe}{1-f_{v}^{(\epsilon)}},\quad \mbox{with probability } 1.
\]
\end{thm}
\begin{proof}

The network $\mathcal{N}_{C}^{(\epsilon)}$ described above is an
open network with Poisson exogenous arrival processes. An LCFS-PR queue
management is used at each node in $\Nc.$ The service requirement
of each flow at each node is deterministic with bounded support. As
shown in \cite{kelly1979reversibility} (see Theorem 3.10 of Chapter 3), the network
$\Nc$ is an open network of quasi-reversible queues. Therefore, the
queue size distribution of $\Nc$ has a product form in equilibrium.
Further, the marginal distribution of each queue in equilibrium is
the same as if the queue is operating in isolation. Consider a node
$v.$ We denote by $Q_{v}$ the queue size of node $v$ in equilibrium.
In isolation, $Q_{v}$ is distributed as if the queue has a Poisson
arrival process of rate $f_{v}^{(\epsilon)}.$ Theorem $3.10$ implies
that the queue is quasi-reversible and hence it has a distribution such
that 
\[
\E[Q_{v}]=\frac{f_{v}^{(\epsilon)}}{1-f_{v}^{(\epsilon)}}.
\]

Let $Q_{v}^{(j,x)}$ denote the number of type-$(j,x)$ flows at node
$v$ in equilibrium. Note that $\xe\lambda_{j,x}$ is the average
amount of service requirement for type-$(j,x)$ flows at node $v$
per unit time. Then by Theorem $3.10$ \cite{kelly1979reversibility}, we have 
\[
\E[Q_{v}^{(j,x)}]=\frac{\xe\lambda_{j,x}}{f_{v}^{(\epsilon)}}\E[Q_{v}]. 
\]

Let $\bar{D}_{v}^{(j,x),\epsilon} (\blambda)$ denote the delay experienced by a type-$(j,x)$ flow at node $v,$ in equilibrium.
Applying Little's Law to the stable system concerning only type-$(j,x)$
flows at node $v$, we can obtain 
\[
\E[\bar{D}_{v}^{(j,x),\epsilon} (\blambda)]=\frac{\E[Q_{v}^{(j,x)}]}{\lambda_{j,x}}=\frac{\xe}{1-f_{v}^{(\epsilon)}}.
\]

Therefore, the delay experienced by a flow of size $x$ along the entire route~$j,$ $\bar{D}^{(j,x),\epsilon}(\blambda),$ satisfies
\[
\E[\bar{D}^{(j,x),\epsilon} (\blambda)]=\sum_{v:v\in j}\E[\bar{D}_{v}^{(j,x),\epsilon} (\blambda)]=\sum_{v:v\in j}\frac{\xe}{1-f_{v}^{(\epsilon)}}.
\]

By the ergodicity of the network $\Nc,$ with probability 1, 
\[
{D}^{(j,x),\epsilon} (\blambda) =\E[\bar{D}^{(j,x),\epsilon} (\blambda)]=\sum_{v:v\in j}\frac{\xe}{1-f_{v}^{(\epsilon)}}.
\]

\end{proof}

\subsection{Scheduling Scheme for $\mathcal{N}_{D}^{(\epsilon)}$\label{subsec:Nd}}

In the discrete-time networks, each flow is packetized and time is
slotted for packetized transmission. The scheduling policy for $\Nd$
differs from that for $\mathcal{N}_{C}^{(\epsilon)}$ in the following
ways: 
\begin{enumerate}
\item[(1)] A flow/packet generated at time $t$ becomes eligible for transmission
only at the $\lceil\frac{t}{\epsilon}\rceil$-th time slot; 
\item[(2)] A complete packet has to be transmitted in a time slot, i.e., fractions
of packets cannot be transmitted. 
\end{enumerate}
Therefore, the LCFS-PR policy in $\Nc$ cannot be directly implemented.
We need to adapt the LCFS-PR policy to our discrete-time setting.
The adaption was first introduced by El Gamal et al.~\cite{gamal2006throughput_delay}
and has been applied to the design of a scheduling algorithm for a
constrained network \cite{jagabathula2008delay_scheduling}. However,
it was restricted to the setup where all flows had unit size. Unlike that, here we
consider variable size flows. Before presenting the adaptation of LCFS-PR scheme, 
we first make the following definition. 
\begin{defn}[Flow Arrival Time in $\mathcal{N}_{D}^{(\epsilon)}$]
The arrival time of a flow at a node $v$ 
in $\mathcal{N}_{D}^{(\epsilon)}$
is defined as the arrival time of its last packet at node
$v$. That is, assume that its $i$-th packet arrives at node $v$
at time $p_{i},$ then the flow is said to arrive at node $v$ at
time $A=\max_{1\leq i\leq k}p_{i}.$ Similarly, we define the departure
time of a flow from a node as the moment when its last packet leaves
this node. 
\end{defn}

\medskip
\noindent\textbf{Emulation scheme. }Suppose that exogenous flows arrive at
$\Nc$ and $\Nd$ simultaneously. In $\Nc$, flows are served at each
node according to the LCFS-PR policy. For $\Nd,$ we consider a scheduling
policy where packets of a flow will not be eligible for transmission
until the full flow arrives. Let the arrival time of a flow at some
node in $\mathcal{N}_{C}$ be $\tau$ and in $\mathcal{N}_{D}^{(\epsilon)}$
at the same node be $A.$ Then packets of this flow are served in
$\mathcal{N}_{D}^{(\epsilon)}$ using an LCFS-PR policy with the arrival
time $\epsilon\lceil\frac{\tau}{\epsilon}\rceil$ instead of $A.$
If multiple flows arrive at a node at the same time, the flow with
the largest arrival time in $\Nc$ is scheduled first. For a flow
with multiple packets, packets are transmitted in an arbitrary
order.

\smallskip

This scheduling policy is feasible in $\Nd$ if and only if each flow
arrives before its scheduled departure time, i.e., $A\leq\epsilon\lceil\frac{\tau}{\epsilon}\rceil$
for every flow at each node. Let $\delta$ and $\Delta$ be the departure
times of a flow from some node in $\mathcal{N}_{C}^{(\epsilon)}$
and $\mathcal{N}_{D}^{(\epsilon)}$, respectively. Since the departure
time of a flow at a node is exactly its arrival time at the next node
on the route, it is sufficient to show that $\Delta\leq\epsilon\lceil\frac{\delta}{\epsilon}\rceil$
for each flow in every busy cycle of each node in $\mathcal{N}_{C}^{(\epsilon)}.$
This will be proved in the following lemma.
\begin{lem}
\label{lemma:NcNd_single_node} Assume that a flow departs from a node
in $\mathcal{N}_{C}^{(\epsilon)}$ and $\mathcal{N}_{D}^{(\epsilon)}$
at times $\delta$ and $\Delta,$ respectively. Then $\Delta\leq\epsilon\lceil\frac{\delta}{\epsilon}\rceil.$ 
\end{lem}
\begin{proof}
	As the underlying graph $\mathcal{G}$ for $\Nc$ and $\Nd$ is a
	DAG, there is a topological ordering of the vertices $\mathcal{V}.$
	Without loss of generosity, assume that $v_{1},\ldots,v_{n}$ is a
	topological order of $\mathcal{G}.$ We prove the statement via induction
	on the index of vertex. 
	
	\textbf{Base case. }We show that the statement holds for node $v_{1},$
	i.e., $\Delta\leq\epsilon\lceil\frac{\delta}{\epsilon}\rceil$ for
	each flow in every busy cycle of node $v_{1}.$ We do so via induction
	on the number $ k $ of flows that contribute to the busy cycle of node $v_{1}$
	in $\Nc.$ Consider a busy cycle consisting of flows numbered $1,\ldots,k$
	with arrivals at times $\tau_{1}\leq\ldots\leq\tau_{k}$ and departures
	at times $\delta_{1},\ldots,\delta_{k}.$ We denote by $c_{i}=(j_{i},x_{i})$
	the type of flow numbered $i.$ Let the arrival times of these flows
	at node $v_{1}$ in $\mathcal{N}_{D}^{(\epsilon)}$ be $A_{1},\ldots,A_{k}$
	and departures at times $\Delta_{1},\ldots,\Delta_{k}.$ As the first
	node in the topological ordering, node $v_{1}$ only has external
	arrivals. Hence $A_{i}=\epsilon\lceil\frac{\tau_{i}}{\epsilon}\rceil$
	for $i=1,\ldots,k$. We need to show that $\Delta_{i}\leq\epsilon\lceil\frac{\delta_{i}}{\epsilon}\rceil,$
	for $i=1,\ldots,k.$ For brevity, we let $S_{i}$ denote the schedule
	time for the $i$-th flow in $\mathcal{N}_{D}^{(\epsilon)}.$ We have
	$S_{i}=\epsilon\lceil\frac{\tau_{i}}{\epsilon}\rceil.$
	
	\textbf{Nested base case. }For $k=1,$ since this busy cycle consists
	of only one flow of type $(j_{1},x_{1}),$ no arrival will occur at
	this node until the departure of this particular flow. In other words,
	for each flow that arrives at node $v_{1}$ after time $\tau_{1},$
	its arrival time, denoted by $\tau_{f},$ should satisfy $\tau_{f}\geq\tau_{1}+x_{1}^{(\epsilon)}.$
	Thus $A_{f}\leq\epsilon\left\lceil \frac{\tau_{f}}{\epsilon}\right\rceil ,$
	and the schedule time for flow $f$ in $\mathcal{N}_{D}^{(\epsilon)}$
	should satisfy 
	\[
	S_{f}=\epsilon\left\lceil \frac{\tau_{f}}{\epsilon}\right\rceil \geq\epsilon\left(\left\lceil \frac{\tau_{1}}{\epsilon}\right\rceil +\frac{x_{1}^{(\epsilon)}}{\epsilon}\right).
	\]
	Hence in $\mathcal{N}_{D}^{(\epsilon)}$, flow $f$ will not become
	eligible for service until node $v_{1}$ transmits all packets of
	the flow 1. Therefore, $\Delta_{1}=S_{1}+x_{1}^{(\epsilon)}=\epsilon\left\lceil \frac{\delta_{1}}{\epsilon}\right\rceil .$
	Thus the induction hypothesis is true for the base case.
	
	\textbf{Nested induction step. }Now assume that the bound $\Delta\leq\epsilon\lceil\frac{\delta}{\epsilon}\rceil$  holds for all busy
	cycles consisting of $k$ flows at $v_{1}$ in $\Nc$. Consider a busy cycle
	of $k+1$ flows.
	
	Note that in $\mathcal{N}_{C}^{(\epsilon)},$ the LCFS-PR service
	policy determines that the first flow of the busy cycle is the last
	to depart. That is, $\delta_{1}=\tau_{1}+\sum_{i=1}^{k+1}x_{i}^{(\epsilon)}.$
	Since flow $i$ is in the same busy cycle of flow $1$ in $\mathcal{N}_{C}^{(\epsilon)},$
	its arrival time should satisfy $\tau_{i}<\tau_{1}+\sum_{l=1}^{i-1}x_{l}^{(\epsilon)},$
	for $i=2,\ldots,k+1.$ Thus $\left\lceil \frac{\tau_{i}}{\epsilon}\right\rceil \leq\left\lceil \frac{\tau_{1}}{\epsilon}\right\rceil +\sum_{l=1}^{i-1}\frac{x_{l}^{(\epsilon)}}{\epsilon},$
	i.e., $S_{i}\leq S_{1}+\sum_{l=1}^{i-1}x_{l}^{(\epsilon)}.$ Let us
	focus on the departure times of these flows in $\Nd.$
	
	\smallskip{}
	
	Case (i): There exists some $i$ such that $S_{i}=S_{1}+\sum_{l=1}^{i-1}x_{l}^{(\epsilon)}.$
	Let $i^{*}$ be the smallest $i$ satisfying this equality. Then under
	the LCFS-PR policy in $\mathcal{N}_{D}^{(\epsilon)},$ the $i^{*}$-th
	flow will be scheduled for service right after the departure of flow
	$1.$ That is, $\Delta_{1}=S_{1}+\sum_{l=1}^{i^{*}-1}x_{l}^{(\epsilon)}=\epsilon\left\lceil \frac{\tau_{1}}{\epsilon}\right\rceil +\sum_{l=1}^{i^{*}-1}x_{l}^{(\epsilon)}\leq\epsilon\left\lceil \frac{\tau_{1}}{\epsilon}\right\rceil +\sum_{l=1}^{k+1}x_{l}^{(\epsilon)}=\epsilon\left\lceil \frac{\delta_{1}}{\epsilon}\right\rceil .$
	The remaining flows numbered $i^{*},\ldots,k+1$ depart exactly as
	if they belong to a busy cycle of $k+2-i^{*}$ flows.
	
	\smallskip{}
	
	Case (ii): For $i=2,\ldots,k+1,$ $S_{i}<S_{1}+\sum_{l=1}^{i-1}x_{l}^{(\epsilon)}.$
	Then with the LCFS-PR policy in $\mathcal{N}_{D}^{(\epsilon)},$ flows
	numbered $2,\ldots,k+1$ would have departure times as if they are
	from a busy cycle of $k$ flows. The service for flow $1$ will be
	resumed after the departure of these $k$ flows. In particular, we
	will show that the resumed service for flow $1$ will not be interrupted
	by any further arrival. Then $\Delta_{1}=S_{1}+\sum_{l=1}^{k+1}x_{l}^{(\epsilon)}=\epsilon\left\lceil \frac{\tau_{1}}{\epsilon}\right\rceil +\sum_{l=1}^{k+1}x_{l}^{(\epsilon)}=\epsilon\left\lceil \frac{\delta_{1}}{\epsilon}\right\rceil .$
	
	Since this particular busy cycle in $\mathcal{N}_{C}^{(\epsilon)}$
	consists of $k+1$ flows, arguing similarly as the base case, each
	flow $f$ that arrives at node $v_{1}$ after this busy cycle should
	satisfy $\tau_{f}\geq\delta_{1}.$ Hence its schedule time $S_{f}$
	in $\mathcal{N}_{D}^{(\epsilon)}$ satisfies $S_{f}=\epsilon\left\lceil \frac{\tau_{f}}{\epsilon}\right\rceil \geq\epsilon\left\lceil \frac{\delta_{1}}{\epsilon}\right\rceil =\epsilon\left\lceil \frac{\tau_{1}}{\epsilon}\right\rceil +\sum_{l=1}^{k+1}x_{l}^{(\epsilon)}=S_{1}+\sum_{l=1}^{k+1}x_{l}^{(\epsilon)}.$
	Therefore, flow $f$ will not be eligible for service until the departure
	of flow $1$ in $\mathcal{N}_{D}^{(\epsilon)}.$ This completes the
	proof of the base case $v_{1}$.
	
	\textbf{Induction step. }Now assume that for each $\tau=1,\cdots,t,$
	$\Delta\leq\epsilon\lceil\frac{\delta}{\epsilon}\rceil$ holds for
	each flow in every busy cycle of node $v_{\tau}.$ We show that this
	holds for node $v_{t+1.}$
	
	We can show that $\Delta\leq\epsilon\lceil\frac{\delta}{\epsilon}\rceil$
	via induction on the number of flows that contribute to the busy cycle
	of node $v_{t+1}$ in $\Nc,$ following exactly the same argument
	for the proof of base case $v_{1}.$ Omitting the details, we point
	to the difference from the proof of nested induction step for $v_1$. 
	
	Note that by the topological ordering, flows arriving at node $v_{t+1}$
	are either external arrivals or departures from nodes $v_{1},\ldots,v_{t}.$
	By the induction hypothesis, the arrival times of each flow at $v_{t+1}$
	in $\Nc$ and $\Nd$, denoted by $\tau$ and $A,$ respectively,
	satisfy $A\leq\epsilon\lceil\frac{\tau}{\epsilon}\rceil.$ 
	
	This completes the proof of this lemma.
\end{proof}

\smallskip

We now prove Theorem \ref{thm:scheduling} by Lemma \ref{lemma:NcNd_single_node}.
\begin{proof}[Proof of Theorem \ref{thm:scheduling}]
Suppose $\mathcal{N}_{D}^{(\epsilon)}$ is operating under the LCFS-PR
scheduling policy using the arrival times in $\mathcal{N}_{C}^{(\epsilon)},$
where LCFS-PR queue management is used at each node. Let $D_S^{(j,x),\epsilon} (\blambda)$
and $D_{C}^{(j,x),\epsilon} (\blambda)$ denote the sample mean of the delay over all type-$(j,x)$
flows in $\mathcal{N}_{D}^{(\epsilon)}$ and $\mathcal{N}_{C}^{(\epsilon)}$
respectively. Then it follows from Lemma~\ref{lemma:NcNd_single_node}
that $D_S^{(j,x),\epsilon} (\blambda) \leq D_{C}^{(j,x),\epsilon} (\blambda).$ 
From Theorem
\ref{thm:Delay_Nc}, we have 
\[
D_{C}^{(j,x),\epsilon} (\blambda) =\sum_{v:v\in j}\frac{\xe}{1-f_{v}^{(\epsilon)}}, \quad \mbox{with probability } 1.
\]
Hence, with probability 1, 
\[
D^{(j,x),\epsilon}_S (\blambda) \leq\sum_{v:v\in j}\frac{\xe}{1-f_{v}^{(\epsilon)}}\leq\frac{C_{0}}{C_{0}-1}\sum_{v:v\in j}\frac{\xe}{1-f_{v}}, \label{eq:delay_Nd}
\]
where the last inequality follows from Eq.~(\ref{eq:node_load_Nd_gap}).

By definition,  we have $\frac{1}{1-f_{v}}\leq\frac{1}{1-\rho_{j}(\boldsymbol{\lambda})}$ for each $v\in j.$
Therefore, with probability 1, 
\begin{align*}
D^{(j,x),\epsilon}_S (\blambda) & \leq\frac{C_{0}}{C_{0}-1}\cdot\frac{\epsilon\left\lceil x/\epsilon\right\rceil d_{j}}{1-\rho_{j}(\boldsymbol{\lambda})} \\
& \leq\frac{C_{0}}{C_{0}-1}\cdot\frac{\epsilon (x/\epsilon + 1) d_{j}}{1-\rho_{j}(\boldsymbol{\lambda})} \\
& \leq\frac{C_{0}}{C_{0}-1}\Big(\frac{xd_{j}}{1-\rho_{j}(\boldsymbol{\lambda})} + \frac{\epsilon d_{j}}{1-\rho_{j}(\boldsymbol{\lambda})}\Big) \\
& \leq \frac{C_{0}}{C_{0}-1}\Big(\frac{xd_{j}}{1-\rho_{j}(\boldsymbol{\lambda})} +  d_{j} \Big),
\end{align*}
where the last inequality holds because $ \epsilon \le 1-\rho_j(\blambda) $ by the definition of $ \epsilon $ in Eq.~(\ref{eq:epsilon}).
This completes the proof of Theorem \ref{thm:scheduling}.
\end{proof}

\begin{rem}
	Consider a continuous-time network $\mathcal{N}_{C}$ with the same
	topology and exogenous arrivals as $\Nd,$ while each type-$(j,x)$
	flow maintains the original size $x.$ Using the same argument employed
	for the proof of Theorem \ref{thm:Delay_Nc}, we can show that the expected
	delay of a type-$(j,x)$ flow in $\mathcal{N_{C}},$ $\E[\bar{D}^{(j,x)} (\blambda)]$,
	is such that $\E[\bar{D}^{(j,x)} (\blambda)]=\sum_{v:v\in j}\frac{x}{1-f_{v}}.$ 
	For the discrete-time network $\Nd,$ from Eq.\ (\ref{eq:delay_Nd}), we have 
	\[
    D^{(j,x),\epsilon}_S (\blambda) \leq \frac{C_{0}}{C_{0}-1}\Big(\sum_{v:v\in j}\frac{x}{1-f_{v}} +  d_{j} \Big),
	\]
	Therefore, $\Nd$ achieves essentially the same flow delay bound as $\mathcal{N}_{C}$.
	
\end{rem}

\section{Discussion\label{sec:discussion}}

We have emphasized the delay performance of our centralized congestion
control and scheduling, but have not discussed the implementation
complexity of this scheme. Recall that our approach consists of two
main sub-routines: (a) Simulation of the bandwidth sharing network
operating under the SFA policy to obtain the injection times for congestion
control; (b) Simulation of the continuous-time network $\Nc$ to determine
the arrival times of scheduling in $\Nd.$ We can see that the overall
complexity of our algorithm is dominated by the computation complexity
of the SFA policy. As one has to compute the sum of exponentially
many terms for every event of exogenous flow arrival and flow injection
into network, the complexity of SFA is exponential in $M,$ the number
of nodes in the network. One future direction is to develop a variant
of the SFA policy with low-complexity. As we mentioned before, there
is a close relationship between the SFA policy and proportional fairness
algorithm that maximizes network utility. In particular, Massouli\'{e}
\cite{massoulie2007fairness} proved that the SFA policy converges
to proportional fairness under the heavy-traffic limit. With this,
one can begin by implementing the proportional fairness. Indeed,
some recent work~\cite{nagaraj2016numfabric,perry17flowtune} has demonstrated the feasibility and efficiency of
applying network utility maximization (NUM) for rate allocation (congestion
control).

Our theoretical results are stated under the assumption of Poisson arrival processes. 
This is in fact not a restriction, as it can be relaxed using a Poissonization argument that has
been applied in previous work \cite{gamal2006throughput_delay,shah2014SFA,jagabathula2008delay_scheduling}.
In particular, incoming flows on each route are passed through a ``regularizer'' that
emits flows according to a Poisson process (with a dummy flow emitted if the regularizer is empty).  
The rate of the regularizer can be chosen to lie between the arrival rate and 
the network capacity. Consequently, the regularized
arrivals are Poisson, whereas the effective gap to the capacity, $1-\rho_j(\blambda),$
is decreased by a constant factor, resulting in an extra additive term of $ O \left( \frac{1}{1-\rho_j(\blambda)} \right)$ in our delay bound.

Our theoretical results suggest that a congestion control policy that is 
proportional fair along with packet scheduling in the internal network that is simple
discrete time LCFS-PR should achieve an excellent performance in practice. And such
an algorithm does not seem to be far from being implementable as exemplified
by \cite{perry2014fastpass, perry17flowtune}.

\bibliographystyle{ACM-Reference-Format}
\bibliography{datacenter}

\end{document}